\documentclass[12pt, a4paper]{article}
\usepackage{amsmath, amssymb, amsthm}
\usepackage{mathtools}
\usepackage{graphics}
\usepackage{graphicx}
\usepackage{epsfig}
\usepackage{dblfloatfix}
\usepackage{float}
\usepackage{placeins}
\usepackage{flafter}
\usepackage[usenames]{color}

\usepackage{epstopdf}
\usepackage{url}
\usepackage{hyperref}
\date{}
\usepackage{fullpage}
\usepackage{setspace}

\newtheorem{df}{Definition}[section]
\newtheorem{pro}{Proposition}[section]
\newtheorem{thm}{Theorem}[section]
\newtheorem{cor}{Corollary}[section]

\newtheorem{exam}{Example}[section]

\makeatletter

\renewcommand\section{\@startsection {section}{1}{\z@}
{-30pt \@plus -1ex \@minus -.2ex} {2.3ex \@plus.2ex}
{\normalfont\normalsize\bfseries}}

\renewcommand\subsection{\@startsection{subsection}{2}{\z@}
{-3.25ex\@plus -1ex \@minus -.2ex} {1.5ex \@plus .2ex}
{\normalfont\normalsize\bfseries}}

\renewcommand{\@seccntformat}[1]{\csname the#1\endcsname. }

\makeatother
\title{\bf{Skew cyclic and skew $\boldsymbol{(\alpha_{1}+u\alpha_{2}+v\alpha_{3}+uv\alpha_{4})}$-constacyclic codes over $\boldsymbol{F_{q}+uF_{q}+vF_{q}+uvF_{q}}$}}
\author{ \bf Habibul Islam and Om Prakash\\\\
Department of Mathematics \\
Indian Institute of Technology Patna\\ Patna- 801 106, India \\
E-mail: habibul.pma17@iitp.ac.in and om@iitp.ac.in}
\begin{document}

\maketitle

\begin{abstract}

In this note, we study skew cyclic and skew constacyclic codes over the ring $\mathcal{R}=F_{q}+uF_{q}+vF_{q}+uvF_{q}$ where $q=p^{m},$ $p$ is an odd prime, $u^{2}=u,~v^{2}=v,~uv=vu$. We show that Gray images of a skew cyclic and skew $\alpha$-constacyclic code of length $n$ are skew quasi-cyclic code of length $4n$ over $F_{q}$ of index 4. Also, it is shown that skew $\alpha$-constacyclic codes are either equivalent to $\alpha$-constacyclic codes or $\alpha$-quasi-twisted codes over $\mathcal{R}$. Further, structural properties, specially, generating polynomials and idempotent generators for skew cyclic and skew constacyclic codes are determined by decomposition method.
\end{abstract}


%%%%%%%%%%%%%%%%%%%%%%%%%%%%%%%%%%%%%%%%%%%%%%%%%%%%%%%%%%%%%%%%%%%%%%%%%%%%%%%%%%%%%%
 \section{Introduction}
Linear codes have been studied in coding theory for last six decades. Initially linear codes were considered and studied over binary field. In 1994, a remarkable work presented  by  Hammons et al. \cite{hamm94}, which shows that some good binary non-linear codes can be considered as linear codes over $\mathbb{Z}_{4}$ through Gray map. Consequently, the study of linear codes over finite rings has got more attention than the study of linear codes over binary field. One of the most important class of linear codes is cyclic code which is easier to implement and also playing crucial role in the development of algebraic coding theory. Latter, cyclic codes over finite rings have got attention of many researchers and hence many new techniques have been discovered to produce cyclic codes over finite commutative ring with better parameters and properties [\cite{tabul03}, \cite{dinh10}, \cite{xiying17}, \cite{joel}]. But all of these works are restricted to the finite commutative rings.\\

In 2007, Boucher et al. \cite{D07} introduced skew cyclic codes which is a generalization of the class of cyclic codes on a non-commutative ring, namely, skew polynomial ring $F[x;\theta]$, where $F$ is a field and $\theta$ is an automorphism on $F$. Boucher and Ulmer \cite{D09} constituted some skew cyclic codes with Hamming distance larger than the previously known linear codes with the same parameters. Later on, Abualrub et al. \cite{abul12} studied skew cyclic codes over $F_{2}+vF_{2}$, where $v^{2}=v$. In these works, they constructed skew cyclic codes by taking an automorphism whose order must divide the length of the code. By relaxing the said restriction, Siap et al. \cite{siap11} investigated structural properties of skew cyclic codes of arbitrary length and found that a skew cyclic code of arbitrary length is equivalent to either a cyclic code or a quasi-cyclic code over $F_{p}+vF_{p}$, where $v^{2}=v$. In 2014, Gursoy et
al. \cite{Gursoy14} constructed a skew cyclic code along with their generators and also found the idempotent generators of skew cyclic codes over $F_{p}+vF_{p}$, where $v^{2}=v$. Jitman
et al. \cite{jitman} introduced skew constacyclic codes over the skew polynomial ring whose coefficients are comes from finite chain rings, particularly for the ring $F_{p^{m}} + uF_{p^{m}}$ where $u^{2} = 0.$ Very recently,  Gao et al. \cite{gao} investigated some structural properties of skew constacyclic codes over the ring $F_{q}+vF_{q}; v^{2}=v$. \\

Inspired by these study, we consider $\mathcal{R}=F_{q}+uF_{q}+vF_{q}+uvF_{q}$, where $q=p^{m},~p$ is an odd prime and $u^{2}=u,v^{2}=v,uv=vu $ and study skew cyclic codes and skew constacyclic codes over $\mathcal{R}$.
Now, for skew polynomial ring, we define an automorphism
\begin{align*}
{\theta}_{t}: \mathcal{R}\rightarrow \mathcal{R}
\end{align*}
by
\begin{align*}
{\theta}_{t}(a+ub+vc+uvd)=a^{p^{t}}+ub^{p^{t}}+vc^{p^{t}}+uvd^{p^{t}},
\end{align*}
where $a,b,c,d\in F_{q}$ and use the same throughout this paper.
In this case, the order of the automorphism is $\mid\langle{\theta}_{t}\rangle\mid =\frac{m}{t}=k$ (say). Also, the invariant subring under the automorphism ${\theta}_{t}$ is $F_{p^{t}}+uF_{p^{t}}+vF_{p^{t}}+uvF_{p^{t}}$. For this automorphism ${\theta}_{t}$ on $\mathcal{R}$, the  set $\mathcal{R}[x;{\theta}_{t}]=\big \{a_{0}+a_{1}x+a_{2}x^{2}+\dots +a_{n}x^{n}:a_{i}\in \mathcal{R}, \forall~ i=0,1,2,...,n \big \}$ forms a ring under the usual addition of polynomials and the multiplication, denoted by $\ast$, with respect to $(ax^{i})\ast(bx^{j})=a{{\theta}_{t}}^{i}(b)x^{i+j}$. This is a non-commutative ring, unless ${\theta}_{t}$ is identity, known as a skew polynomial ring. \\

Presentation of the manuscript is as follows: In section \ref{sec2}, we discuss some preliminaries which have been used later. In section \ref{sec3}, we define Gray map and describe Gray images of skew cyclic codes over $\mathcal{R}$. Section \ref{sec4} contains some important results of skew cyclic codes over $\mathcal{R}$ while section \ref{sec5} gives us the idempotent generators of skew cyclic along with their duals and includes two examples to elaborate the derived results. In section \ref{sec6} presents a relation among skew constacyclic, constacyclic and quasi-twisted code.  Section \ref{sec7} discusses the characterization of skew $\alpha$-constacyclic code along with quasi-twisted shift operator. Section \ref{sec8} gives structure of skew constacyclic code by decomposition method and section \ref{sec9} concludes the paper.

\section{Preliminaries}\label{sec2}

Throughout the paper $\mathcal{R}$ represents $F_{q}+uF_{q}+vF_{q}+uvF_{q}$, where $q=p^{m}, ~p$ is an odd prime and $u^{2}=u,v^{2}=v,uv=vu$. It can be checked that $\mathcal{R}$ is a finite commutative ring containing $q^{4}$ elements and characteristic $p$.\\  We recall that a linear code $\mathcal{C}$ of length $n$  over $\mathcal{R}$ is a $\mathcal{R}-$submodule of $\mathcal{R}^{n}$ and members of $\mathcal{C}$ are called codewords. A linear code $\mathcal{C}$ of length $n$ over $\mathcal{R}$ is said to be a skew cyclic code with respect to the automorphism  ${\theta}_{t}$ if and only if  for any codeword $a=(a_{1},a_{2},...,a_{n})\in \mathcal{C}$ implies $\sigma(a)=( {\theta}_{t}(a_{n}),{\theta}_{t}(a_{0}),...,{\theta}_{t}(a_{n-1}) )\in \mathcal{C}$, where $\sigma$ is a skew cyclic shift of $\mathcal{C}$.
\\The inner product of any $a=(a_{1},a_{2},...,a_{n})$, $b=(b_{1},b_{2},...,b_{n})$ in $\mathcal{C}$ is define as $a.b=\sum_{i=1}^{n}a_{i}b_{i}$ and $a,b$ is said to be orthogonal if $a.b=0$. For a code $\mathcal{C}$, its dual denoted by $\mathcal{C}^{\perp}$ and define as $\mathcal{C}^{\perp}=\big \{ x\in \mathcal{R}^{n}:x.c=0 ~\forall c\in \mathcal{C}\big \}$. If $\mathcal{C}=\mathcal{C}^{\perp}$, then $\mathcal{C}$ is said to be a self-dual code.
The Hamming weight $w_{H}(a)$ is defined as the number of the non-zero components in $a=(a_{1},a_{2},...,a_{n})\in \mathcal{C}$ and Hamming distance between two codewords $a$ and $b$ of $\mathcal{C}$ is defined as $d_{H}(a,b)=w_{H}(a-b)$ while the Hamming distance for the code $\mathcal{C}$ is denoted by $d_{H}(\mathcal{C})$ and defined as $d_{H}(\mathcal{C})=min\big\{ d_{H}(a,b)\mid a\neq b, ~\forall a,b \in \mathcal{C} \big\}$.
The Lee weight of an element $r=a+ub+vc+uvd\in \mathcal{R}$ is defined by
$w_{L}(r)= w_{H}(a,a+b,a+c,a+b+c+d)$ and Lee weight for the codeword $a=(a_{1},a_{2},...,a_{n})\in \mathcal{R}^{n}$ is $w_{L}(a)=\sum_{i=1}^{n}w_{L}(a_{i}).$ The Lee distance between two codewords $a, b\in R^{n}$ is defined as $d_{L}(a,b)= w_{L}(a-b)=\sum_{i=1}^{n}w_{L}(a_{i}-b_{i})$ and the Lee distance for the code $\mathcal{C}$ is defined by $d_{L}(\mathcal{C})=min\big\{ d_{L}(a,b)\mid a\neq b, \forall a,b\in \mathcal{C}\big \}$.\\
A code $\mathcal{C}$ of length $nm$ over $F_{q}$ is said to be a skew quasi-cyclic code of index $m$ if ${\pi}_{m}(\mathcal{C})=\mathcal{C}$, where ${\pi}_{m}$ is the skew quasi-cyclic shift on $(F_{q}^{n})^{m}$ define by
\begin{align}\label{eq 1}
{\pi}_{m}(a^{1}\mid a^{2}\mid...\mid a^{m})=(\sigma(a^{1})\mid \sigma(a^{2})\mid...\mid \sigma(a^{m}) ).
\end{align}

\section{Gray map and $F_{q}$ -images of skew cyclic code over $\mathcal{R}$}\label{sec3}

A map $\Psi: \mathcal{R} \rightarrow F_{q}^{4}$ defined by
\begin{align*}
\Psi(a+ub+vc+uvd)=(a,a+b,a+c,a+b+c+d)
\end{align*}
is called the Gray map [\cite{xiying17}]. It can be checked that $\Psi$ is a linear map and can be extended to $\mathcal{R}^{n}$ in obvious way by $\Psi: \mathcal{R}^{n}\rightarrow F_{q}^{4n}$ such that
\begin{align*}
\Psi(r_{1},r_{2},...,r_{n})&=(a_{1},...,a_{n},a_{1}+b_{1},...,a_{n}+b_{n},a_{1}+c_{1},...,a_{n}+c_{n},\\
&a_{1}+b_{1}+c_{1}+d_{1},...,a_{n}+b_{n}+c_{n}+d_{n})
\end{align*}
where $r_{i}=a_{i}+ub_{i}+vc_{i}+uvd_{i}\in \mathcal{R} ~\forall~ i=1,2,...,n.$\\
In light of above definition, the following fact can be easily checked:\\

\begin{thm}\label{Gray th1}
The Gray map $\Psi$ is a $F_{q}$-linear distance preserving map from $\mathcal{R}^{n}$ (Lee distance) to $F_{q}^{4n}$ (Hamming distance).
\end{thm}

\begin{thm}\label{Gray th2}
If $\mathcal{C}$ is $[n,k,d_{L}]$ linear code over $\mathcal{R}$, then $\Psi(\mathcal{C})$ is $[4n,k,d_{H}]$ linear code over $F_{q}$, where $d_{L}=d_{H}.$
\end{thm}

\begin{proof}
By Theorem \ref{Gray th1}, $\Psi$ is a $F_{q}$-linear map, so $\Psi(\mathcal{C})$ is a linear code of length $4n$. Since $\Psi$ is bijection as well as distance preserving, $\Psi(\mathcal{C})$ has same minimum distance as $\mathcal{C}$ $i.e$, $d_{L}=d_{H}$ and  dimension $k$.
\end{proof}

\begin{pro}\label{pro1}
Let $\sigma$ be the skew cyclic shift and ${\pi}_{4}$ the skew quasi-cyclic shift as defined in equation \ref{eq 1} and $\Psi$ be the Gray map from $\mathcal{R}^{n}$ to $F_{q}^{4n}$. Then $\Psi \sigma={\pi}_{4}\Psi$.
\end{pro}

\begin{proof}
Let $r_{i}=a_{i}+ub_{i}+vc_{i}+uvd_{i}\in \mathcal{R}$ for $i=1,2,..., n.$ Take $r=(r_{1},r_{2},\dots ,r_{n})$, then $\sigma(r)=( {\theta}_{t}(r_{n}),{\theta}_{t}(r_{1}),...,{\theta}_{t}(r_{n-1})).$ Applying $\Psi$, we have $\Psi(\sigma(r))=( a^{p^{t}}_{n},{a^{p^{t}}_{1}},...,{a^{p^{t}}_{n-1}},{a^{p^{t}}_{n}}+{b^{p^{t}}_{n}},...,{a^{p^{t}}_{n-1}}+{b^{p^{t}}_{n-1}},{a^{p^{t}}_{n}}+{c^{p^{t}}_{n}},...,{a^{p^{t}}_{n-1}}+{c^{p^{t}}_{n-1}},{a^{p^{t}}_{n}}+{b^{p^{t}}_{n}}+{c^{p^{t}}_{n}}+{d^{p^{t}}_{n}},...,{a^{p^{t}}_{n-1}}+{b^{p^{t}}_{n-1}}+{c^{p^{t}}_{n-1}}+{d^{p^{t}}_{n-1}})$.
On the other side, $\Psi(r) =(a_{1},a_{2},...,a_{n},a_{1}+b_{1},...,a_{n}+b_{n},a_{1}+c_{1},...,a_{n}+c_{n},...,a_{1}+b_{1}+c_{1}+d_{1},...,a_{n}+b_{n}+c_{n}+d_{n})$. Therefore, ${\pi}_{4}(\Psi(r))=( {a^{p^{t}}_{n}},{a^{p^{t}}_{1}},...,{a^{p^{t}}_{n-1}},{a^{p^{t}}_{n}}+{b^{p^{t}}_{n}},...,{a^{p^{t}}_{n-1}}+{b^{p^{t}}_{n-1}},{a^{p^{t}}_{n}}+{c^{p^{t}}_{n}},...,{a^{p^{t}}_{n-1}}+{c^{p^{t}}_{n-1}},{a^{p^{t}}_{n}}+{b^{p^{t}}_{n}}+{c^{p^{t}}_{n}}+{d^{p^{t}}_{n}},...,{a^{p^{t}}_{n-1}}+{b^{p^{t}}_{n-1}}+{c^{p^{t}}_{n-1}}+{d^{p^{t}}_{n-1}})$. Hence $\Psi \sigma={\pi}_{4}\Psi$.

\end{proof}

\begin{thm}
A linear code $\mathcal{C}$ of length $n$ over $\mathcal{R}$ is a skew cyclic if and only if $\Psi(\mathcal{C})$ is a skew quasi-cyclic code of length $4n$ over $F_{q}$ of index 4.
\end{thm}
\begin{proof}
Let $\mathcal{C}$ be a skew cyclic code of length $n$ over $\mathcal{R}$. So, $\sigma(\mathcal{C})=\mathcal{C}$ and this implies $\Psi(\sigma(\mathcal{C}))=\Psi(\mathcal{C})$. By Proposition \ref{pro1}, we have ${\pi}_{4}(\Psi(\mathcal{C}))=\Psi(\mathcal{C})$. This shows that $\Psi(\mathcal{C})$ is a skew quasi-cyclic code of length $4n$ over $F_{q}$ of index 4.\\
Conversely, suppose $\Psi(\mathcal{C})$ is a skew quasi-cyclic code of length $4n$ over $F_{q}$ of index 4. Then ${\pi}_{4}(\Psi(\mathcal{C}))=\Psi(\mathcal{C})$. Proposition \ref{pro1}, we get $\Psi(\sigma(\mathcal{C}))=\Psi(\mathcal{C})$ and since $\Psi$ is injective, so $\sigma(\mathcal{C})=\mathcal{C}$, $i.e$, $\mathcal{C}$ is a skew cyclic code over $\mathcal{R}$.
\end{proof}

It is noted that sometimes the use of permuted version of Gray map instead of using direct Gray map is more convenient. The permuted version of $\Psi_{\pi}$ defined as $\Psi_{\pi}(r) = (a_{0}, a_{0}+b_{0}, a_{0}+c_{0}, a_{0}+b_{0}+c_{0}+d_{0}, a_{1}, a_{1}+b_{1}, a_{1}+c_{1}, a_{1}+b_{1}+c_{1}+d_{1},\dots , a_{n-1}, a_{n-1}+b_{n-1}, a_{n-1}+c_{n-1}, a_{n-1}+b_{n-1}+c_{n-1}+d_{n-1})$. Clearly, the codes obtained by $\Psi$ and $\Psi_{\pi}$ are permutation equivalent. Now, for a particular case, $\mid\langle\theta_{t}\rangle\mid = k = 3$, we have the following result:

\begin{pro}
For any $r\in \mathcal{R}^{n}$, we have $\Psi_{\pi}\sigma(r) = \sigma^{4}\Psi_{\pi}(r)$.
\end{pro}
\begin{proof}
Let $r = (r_{0}, r_{1},\dots , r_{n-1})\in \mathcal{R}^{n}$ where $r_{i}=a_{i}+ub_{i}+vc_{i}+uvd_{i}$ for $0\leq i\leq n-1.$ We have,
$\Psi_{\pi}(\sigma(r)) = \Psi_{\pi}(\theta_{t}(r_{n-1}), \theta_{t}(r_{0}), \dots , \theta_{t}(r_{n-2}))=( a^{p^{t}}_{n-1}, a^{p^{t}}_{n-1}+b^{p^{t}}_{n-1}, a^{p^{t}}_{n-1}+c^{p^{t}}_{n-1}, a^{p^{t}}_{n-1}+b^{p^{t}}_{n-1}+c^{p^{t}}_{n-1}+d^{p^{t}}_{n-1}, a^{p^{t}}_{0}, a^{p^{t}}_{0}+b^{p^{t}}_{0}, a^{p^{t}}_{0}+c^{p^{t}}_{0}, a^{p^{t}}_{0}+b^{p^{t}}_{0}+c^{p^{t}}_{0}+d^{p^{t}}_{0}, \dots , a^{p^{t}}_{n-2}, a^{p^{t}}_{n-2}+b^{p^{t}}_{n-2}, a^{p^{t}}_{n-2}+c^{p^{t}}_{n-2}, a^{p^{t}}_{n-2}+b^{p^{t}}_{n-2}+c^{p^{t}}_{n-2}+d^{p^{t}}_{n-2}).$\\ \\
On the other hand,\\ \\
$\sigma^{4}\Psi_{\pi}(r) = \sigma^{4}(a_{0}, a_{0}+b_{0}, a_{0}+c_{0}, a_{0}+b_{0}+c_{0}+d_{0}, a_{1}, a_{1}+b_{1}, a_{1}+c_{1}, a_{1}+b_{1}+c_{1}+d_{1}, \dots , a_{n-1}, a_{n-1}+b_{n-1}, a_{n-1}+c_{n-1}, a_{n-1}+b_{n-1}+c_{n-1}+d_{n-1}) =(\theta_{t}^{4}(a_{n-1}), \theta_{t}^{4}(a_{n-1}+b_{n-1}), \theta_{t}^{4}(a_{n-1}+c_{n-1}), \theta_{t}^{4}(a_{n-1}+b_{n-1}+c_{n-1}+d_{n-1}), \theta_{t}^{4}(a_{0}), \theta_{t}^{4}(a_{0}+b_{0}), \theta_{t}^{4}(a_{0}+c_{0}), \theta_{t}^{4}(a_{0}+b_{0}+c_{0}+d_{0}), \dots \theta_{t}^{4}(a_{n-2}), \theta_{t}^{4}(a_{n-2}+b_{n-2}), \theta_{t}^{4}(a_{n-2}+c_{n-2}), \theta_{t}^{4}(a_{n-2}+b_{n-2}+c_{n-2}+d_{n-2})).$\\ \\
As $\mid\langle \theta_{t}\rangle\mid = 3$, $\theta_{t}^{3}(a) = a, ~~\forall~~ a\in \mathcal{R}$. Finally,\\ \\
$\sigma^{4}\Psi_{\pi}(r) = ( a^{p^{t}}_{n-1}, a^{p^{t}}_{n-1}+b^{p^{t}}_{n-1}, a^{p^{t}}_{n-1}+c^{p^{t}}_{n-1}, a^{p^{t}}_{n-1}+b^{p^{t}}_{n-1}+c^{p^{t}}_{n-1}+d^{p^{t}}_{n-1}, a^{p^{t}}_{0}, a^{p^{t}}_{0}+b^{p^{t}}_{0}, a^{p^{t}}_{0}+c^{p^{t}}_{0}, a^{p^{t}}_{0}+b^{p^{t}}_{0}+c^{p^{t}}_{0}+d^{p^{t}}_{0},\dots , a^{p^{t}}_{n-2}, a^{p^{t}}_{n-2}+b^{p^{t}}_{n-2}, a^{p^{t}}_{n-2}+c^{p^{t}}_{n-2}, a^{p^{t}}_{n-2}+b^{p^{t}}_{n-2}+c^{p^{t}}_{n-2}+d^{p^{t}}_{n-2}).$
This completes the proof.
\end{proof}

In the light of the above proposition, we conclude the following:

\begin{thm}
Let $\mid\langle \theta\rangle\mid = k = 3$ and $C$ be a skew cyclic code of length $n$ over $\mathcal{R}$. Then its $F_{q}$ -image $\Psi_{\pi}(C)$ is equivalent to a $4-$quasicyclic code of length $4n$ over $F_{q}$.
\end{thm}

\section{Skew cyclic codes over $\mathcal{R}$}\label{sec4}

 We denote $B_{1}\oplus B_{2}\oplus B_{3}\oplus B_{4}=\big\{ a_{1}+a_{2}+a_{3}+a_{4}:a_{i}\in B_{i}$ for $i=1,2,3,4\big \}$.\\
Also, any $a+ub+vc+uvd\in \mathcal{R}$ can be written as $a+ub+vc+uvd=(1-u-v+uv)a+(u-uv)(a+b)+(v-uv)(a+c)+uv(a+b+c+d)$ where $a,b,c,d\in F_{q}$ and by \cite{joel}, $a+ub+vc+uvd\in \mathcal{R}$ is a unit if and only if $a,(a+b),(a+c),(a+b+c+d)$ are units in $F_{q}$. Let $\mathcal{C}$ be a linear code of length $n$ in $\mathcal{R}$ and \\
$\mathcal{C}_{1}=\big \{a\in F_{q}^{n}\mid a+ub+vc+uvd\in \mathcal{C}$, ~for ~some~ $b,c,d\in F_{q}^{n} \big\}$;\\
$\mathcal{C}_{2}=\big \{a+b\in F_{q}^{n}\mid a+ub+vc+uvd\in \mathcal{C}$, ~for ~some~ $c,d\in F_{q}^{n} \big\}$;\\
$\mathcal{C}_{3}=\big \{a+c\in F_{q}^{n}\mid a+ub+vc+uvd\in \mathcal{C}$, ~for ~some~ $b,d\in F_{q}^{n} \big\}$;\\
$\mathcal{C}_{4}=\big \{a+b+c+d\in F_{q}^{n}\mid a+ub+vc+uvd\in \mathcal{C}\big\}$.\\
Then $\mathcal{C}_{1},\mathcal{C}_{2},\mathcal{C}_{3},\mathcal{C}_{4}$ are linear codes of length $n$ over ${F_{q}}$ and $\mathcal{C}$ can be expressed uniquely as $\mathcal{C}=(1-u-v+uv)\mathcal{C}_{1}\oplus(u-uv)\mathcal{C}_{2}\oplus(v-uv)\mathcal{C}_{3}\oplus uv\mathcal{C}_{4}$.\\ \\
Now, we calculate dual code of $\mathcal{C}$, denoted by $\mathcal{C}^{\perp}$, as follows:

\begin{thm}\label{dual th}
 Let $\mathcal{C}=(1-u-v+uv)\mathcal{C}_{1}\oplus(u-uv)\mathcal{C}_{2}\oplus(v-uv)\mathcal{C}_{3}\oplus uv\mathcal{C}_{4}$ be a linear code over $\mathcal{R}$ where $\mathcal{C}_{1},\mathcal{C}_{2},\mathcal{C}_{3}$ and $\mathcal{C}_{4}$ be linear codes over ${F_{q}}$. Then $\mathcal{C}^{\perp}=(1-u-v+uv)\mathcal{C}_{1}^{\perp}\oplus(u-uv)\mathcal{C}_{2}^{\perp}\oplus(v-uv)\mathcal{C}_{3}^{\perp}\oplus uv\mathcal{C}_{4}^{\perp}$. Further, $\mathcal{C}$ is self-dual if and only if $\mathcal{C}_{1},\mathcal{C}_{2},\mathcal{C}_{3}$ and $\mathcal{C}_{4}$ are self-dual.
\end{thm}
\begin{proof}
The proof is similar to the proof of Theorem 5 in \cite{xiying17}.
\end{proof}

Let $\mathcal{D}$ be a linear code of length $n$ over $F_{q}$. For any codeword $c=(c_{0},c_{1},...,c_{n-1})$ in $\mathcal{D}$, we identify to a polynomial $c(x)$ in ${F_{q}[x;{\theta}_{t}]}/{\langle x^{n}-1\rangle}$ where $c(x)=c_{0}+c_{1}x+\dots +c_{n-1}x^{n-1}.$ By this identification   Siap et al.  \cite{siap11} studied skew cyclic codes over the field $F_{q}$ and constructed many structural properties like $\mathcal{D}$ is a skew cyclic code over $F_{q}$ if and only if $\mathcal{D}$ is a left $F_{q}[x;{\theta}_{t}]-$submodule of  ${F_{q}[x;{\theta}_{t}]}/{\langle x^{n}-1\rangle}$
and $\mathcal{D}$ is generated by a monic polynomial which is a right divisor of $(x^{n}-1)$ in $F_{q}[x;{\theta}_{t}]$. Analogous to their study we discuss skew cyclic codes over $\mathcal{R}$ in this article.
\begin{thm}\label{sk th1}
Let $\mathcal{C}=(1-u-v+uv)\mathcal{C}_{1}\oplus(u-uv)\mathcal{C}_{2}\oplus(v-uv)\mathcal{C}_{3}\oplus uv\mathcal{C}_{4}$ be a linear code of length $n$ over $\mathcal{R}$ where $\mathcal{C}_{1},\mathcal{C}_{2},\mathcal{C}_{3}$ and $\mathcal{C}_{4}$ be linear codes of length $n$ over ${F_{q}}$. Then $\mathcal{C}$ is a skew cyclic code over $\mathcal{R}$ if and only if $\mathcal{C}_{1},\mathcal{C}_{2},\mathcal{C}_{3}$ and $\mathcal{C}_{4}$ are skew cyclic codes over ${F_{q}}$.
\end{thm}
\begin{proof}
Suppose $\mathcal{C}_{1},\mathcal{C}_{2},\mathcal{C}_{3}$ and $\mathcal{C}_{4}$ are skew cyclic codes over ${F_{q}}$. In order to prove $\mathcal{C}$ is a skew cyclic code over $\mathcal{R}$, let $r=(r_{0},r_{1},...,r_{n-1})\in \mathcal{C}$ where $r_{i}=(1-u-v+uv)a_{i}+(u-uv)b_{i}+(v-uv)c_{i}+uvd_{i}, 0\leq i\leq n-1.$  Take $a=(a_{0},a_{1},...,a_{n-1}), ~ b=(b_{0},b_{1},...,b_{n-1}), ~c=(c_{0},c_{1},...,c_{n-1}), ~d=(d_{0},d_{1},...,d_{n-1})$. Then $a\in \mathcal{C}_{1},b\in \mathcal{C}_{2},c\in \mathcal{C}_{3},d\in \mathcal{C}_{4}$ and since $\mathcal{C}_{1},\mathcal{C}_{2},\mathcal{C}_{3}$ and $\mathcal{C}_{4}$ are skew cyclic, so $\sigma(a)\in \mathcal{C}_{1},\sigma(b)\in \mathcal{C}_{2},\sigma(c)\in \mathcal{C}_{3},\sigma(d)\in \mathcal{C}_{4}.$ Note that ${\theta}_{t}(r_{i})=(1-u-v+uv){\theta}_{t}(a_{i})+(u-uv){\theta}_{t}(b_{i})+(v-uv){\theta}_{t}(c_{i})+uv{\theta}_{t}(d_{i})$ for $0\leq i\leq n-1.$ Then  $\sigma(r)=({\theta}_{t}(r_{n-1}),{\theta}_{t}(r_{0}),...{\theta}_{t}(r_{n-2})=(1-u-v+uv)\sigma(a)+(u-uv)\sigma(b)+(v-uv)\sigma(c)+uv\sigma(d)\in (1-u-v+uv)\mathcal{C}_{1}\oplus(u-uv)\mathcal{C}_{2}\oplus(v-uv)\mathcal{C}_{3}\oplus uv\mathcal{C}_{4}=\mathcal{C}$. This shows that $\mathcal{C}$ is a skew cyclic code over $\mathcal{R}$.\\
Conversely, suppose $\mathcal{C}$ is a skew cyclic code over $\mathcal{R}$. Let $a=(a_{0}, a_{1},.., a_{n-1})\in \mathcal{C}_{1}, b = (b_{0},b_{1},..,b_{n-1})\in \mathcal{C}_{2}, c=(c_{0},c_{1},...,c_{n-1})\in \mathcal{C}_{3}, d=(d_{0},d_{1},..,d_{n-1})\in \mathcal{C}_{4}$. Consider $r_{i}=(1-u-v+uv)a_{i}+(u-uv)b_{i}+(v-uv)c_{i}+uvd_{i}, 0\leq i\leq n-1.$ Then $r=(r_{0},r_{1},...,r_{n-1})\in \mathcal{C}$ and since $\mathcal{C}$ is the skew cyclic so $\sigma(r)\in \mathcal{C}$. But  $\sigma(r)=(1-u-v+uv)\sigma(a)+(u-uv)\sigma(b)+(v-uv)\sigma(c)+uv\sigma(d)$, it follows that $\sigma(a)\in \mathcal{C}_{1},\sigma(b)\in \mathcal{C}_{2},\sigma(c)\in \mathcal{C}_{3},\sigma(d)\in \mathcal{C}_{4}$. Hence $\mathcal{C}_{1},\mathcal{C}_{2},\mathcal{C}_{3}$ and $\mathcal{C}_{4}$ are skew cyclic codes over $F_{q}$.
\end{proof}

\begin{cor}
The  dual code $\mathcal{C}^{\perp}$ is a skew cyclic code over $\mathcal{R}$, provided $\mathcal{C}$ is a skew cyclic code over $\mathcal{R}$.
\end{cor}
\begin{proof}
Let $\mathcal{C}=(1-u-v+uv)\mathcal{C}_{1}\oplus(u-uv)\mathcal{C}_{2}\oplus(v-uv)\mathcal{C}_{3}\oplus uv\mathcal{C}_{4}$ be a skew cyclic code over $\mathcal{R}$. Then by Theorem \ref{sk th1}, $\mathcal{C}_{1},\mathcal{C}_{2},\mathcal{C}_{3}$ and $\mathcal{C}_{4}$ are skew cyclic codes over ${F_{q}}$. Since dual of skew cyclic code over $F_{q}$ is also a skew cyclic code over ${F_{q}}$ by Corollary 18 of \cite{D09}, so $\mathcal{C}_{1}^{\perp},\mathcal{C}_{2}^{\perp},\mathcal{C}_{3}^{\perp}$ and $\mathcal{C}_{4}^{\perp}$ are skew cyclic codes over  ${F_{q}}$. Thus, by Theorem \ref{dual th} and Theorem \ref{sk th1}, $\mathcal{C}^{\perp}$ is a skew cyclic codes over $\mathcal{R}$.
\end{proof}

\begin{cor}
The code $\mathcal{C}$ is a self-dual skew cyclic code over $\mathcal{R}$ if and only if $\mathcal{C}_{1},\mathcal{C}_{2},\mathcal{C}_{3}$ and $\mathcal{C}_{4}$ are self-dual skew cyclic codes over ${F_{q}}$.
\end{cor}

\begin{thm}\label{sk th2}
Let $\mathcal{C}=(1-u-v+uv)\mathcal{C}_{1}\oplus(u-uv)\mathcal{C}_{2}\oplus(v-uv)\mathcal{C}_{3}\oplus uv\mathcal{C}_{4}$ be a skew cyclic code of length $n$ over $\mathcal{R}$. Then $\mathcal{C}$ has a generating polynomial $f(x)$ which is a right divisor of $(x^{n}-1)$ in ${\mathcal{R}[x;{\theta}_{t}]}$.
\end{thm}
\begin{proof}
Let $f_{i}(x)$ be generator of $\mathcal{C}_{i}$ in ${F_{q}[x;{\theta}_{t}]}$ for $i=1,2,3,4.$ Then $(1-u-v+uv)f_{1}(x),(u-uv)f_{2}(x),(v-uv)f_{3}(x),uvf_{4}(x)$ are generators of $\mathcal{C}$. Let $f(x)=(1-u-v+uv)f_{1}(x)+(u-uv)f_{2}(x)+(v-uv)f_{3}(x)+uvf_{4}(x)$ and $\mathcal{G}=\langle f(x) \rangle.$ Then $\mathcal{G} \subseteq \mathcal{C}$. Now, $(1-u-v+uv)f(x) = (1-u-v+uv)f_{1}(x)\in \mathcal{G}, (u-uv)f(x)=(u-uv)f_{2}(x)\in \mathcal{G}, (v-uv)f(x)=(v-uv)f_{3}(x)\in \mathcal{G}, uv f(x) = uv f_{4}(x)\in \mathcal{G}$, it follows that $\mathcal{C}\subseteq \mathcal{G}$ and hence $\mathcal{C}=\mathcal{G}=\langle f(x)\rangle.$\\
Since $f_{i}(x)$ is a right divisor of $(x^{n}-1)$ in ${F_{q}[x;{\theta}_{t}]}$ for $i=1,2,3,4$, so there exit $h_{i}(x)\in {F_{q}[x;{\theta}_{t}]}$ such that $(x^{n}-1)=h_{1}(x)\ast f_{1}(x),(x^{n}-1)=h_{2}(x)\ast f_{2}(x),(x^{n}-1)=h_{3}(x)\ast f_{3}(x),(x^{n}-1)=h_{4}(x)\ast f_{4}(x).$ Now, $[(1-u-v+uv)h_{1}(x)+(u-uv)h_{2}(x)+(v-uv)h_{2}(x)+uvh_{4}(x)]\ast f(x)=(1-u-v+uv)h_{1}(x)\ast f_{1}(x) +(u-uv)h_{2}(x)\ast f_{2}(x)+(v-uv)h_{3}(x)\ast f_{3}(x) +uvh_{4}(x)\ast f_{4}(x)=(x^{n}-1)$. Thus, $f(x)$ is a right divisor of $(x^{n}-1)$.
\end{proof}

\begin{cor}
Each left submodule of $\mathcal{R}[x;\theta_{t}]/\langle x^{n}-1 \rangle$ is generated by single element.
\end{cor}
Let $\mathcal{C}$ be a skew cyclic code over $F_{q}$ generated by the polynomial $f(x)$ such that $(x^{n}-1)= h(x)\ast f(x)$ where $f(x)=f_{0}+f_{1}x+\dots +f_{r}x^{r}, h(x)=h_{0}+h_{1}x+\dots +h_{n-r}x^{n-r}$ in $F_{q}[x;{\theta}_{t}]$. Then by \cite{D09}, its dual $\mathcal{C}^{\perp}$ is a skew cyclic code over $F_{q}$ generated by the polynomial $\widehat{h}(x)=h_{n-r}+{\theta}_{t}(h_{n-r-1})x+\dots +{\theta}_{t}(h_{0})x^{n-r}$.\\

\begin{cor}\label{cor 1}
Let $\mathcal{C}=(1-u-v+uv)\mathcal{C}_{1}\oplus(u-uv)\mathcal{C}_{2}\oplus(v-uv)\mathcal{C}_{3}\oplus uv\mathcal{C}_{4}$ be a skew cyclic code over $\mathcal{R}$ and $f_{i}$ the generator of $\mathcal{C}_{i}$ such that $(x^{n}-1)=h_{i}(x)\ast f_{i}(x)$, for $i=1,2,3,4$. Then $\mathcal{C}^{\perp}=\langle (1-u-v+uv)\widehat{h}_{1}(x)+(u-uv)\widehat{h}_{2}(x)+(v-uv)\widehat{h}_{3}(x)+uv\widehat{h}_{4}(x)\rangle$ and $\mid \mathcal{C}^{\perp}\mid =q^{\sum_{i=1}^{4}f_{i}(x)}$.
\end{cor}
\begin{proof}
Since $\widehat{h}_{i}(x)$ is generator of ${\mathcal{C}_{i}^{\perp}}$ for $i=1,2,3,4$, then by Theorem \ref{dual th} and Theorem \ref{sk th2}, $\mathcal{C}^{\perp}=\langle (1-u-v+uv)\widehat{h}_{1}(x)+(u-uv)\widehat{h}_{2}(x)+(v-uv)\widehat{h}_{3}(x)+uv\widehat{h}_{4}(x)\rangle$.\\
 Also, $\mid  \mathcal{C}^{\perp}\mid=\mid \mathcal{C}_{1}^{\perp}\mid \mid \mathcal{C}_{2}^{\perp}\mid \mid \mathcal{C}_{3}^{\perp}\mid \mid \mathcal{C}_{4}^{\perp}\mid=q^{\sum_{i=1}^{4}f_{i}(x)}.$
\end{proof}

\section{Idempotent generators of skew cyclic codes and their dual codes over $\mathcal{R}$}\label{sec5}

Since underlying ring is the skew polynomial ring which is non-commutative, so polynomials exhibit here more factorizations. Therefore, it is not too easy to find exact number of skew cyclic codes over $\mathcal{R}[x;\theta_{t}]$ or the number of idempotent generators of skew cyclic codes over $\mathcal{R}$. But, when we impose conditions $gcd(n,k)=1$ and $gcd(n,q)=1$ [ as given by \cite{Gursoy14}], where $k$ is the order of the automorphism and $n$ is the length of the code, we can find idempotent generators. Towards this, we have the following:

\begin{thm} (\cite{Gursoy14}) \label{idm th1}
~~Let $f(x)\in F_{q}[x;{\theta}_{t}]$ be a monic right divisor of $(x^{n}-1)$ and $\mathcal{C}=\langle f(x) \rangle.$ If $gcd(n,k)=1$ and $gcd(n,q)=1$, then there exits an idempotent polynomial $e(x)\in F_{q}[x;{\theta}_{t}]/ \langle x^{n}-1 \rangle$ such that $\mathcal{C}=\langle e(x) \rangle.$
\end{thm}

\begin{thm}
Let $\mathcal{C}=(1-u-v+uv)\mathcal{C}_{1}\oplus(u-uv)\mathcal{C}_{2}\oplus(v-uv)\mathcal{C}_{3}\oplus uv\mathcal{C}_{4}$ be a skew cyclic code of length $n$ over $\mathcal{R}$ with $gcd(n,k)=1$ and $gcd(n,q)=1$. Then $\mathcal{C}$ has an idempotent generator $e(x)$ in $R[x;{\theta}_{t}]$.
\end{thm}
\begin{proof}
By Theorem \ref{idm th1}, there exists idempotent generator $e_{i}(x)$ of $\mathcal{C}_{i}$ for $i=1,2,3,4$ in $F_{q}[x;{\theta}_{t}]$. Then by Theorem \ref{sk th2}, $e(x)=(1-u-v+uv)e_{1}(x)+(u-uv)e_{2}(x)+(v-uv)e_{3}(x)+uve_{4}(x)$ is a generator of $\mathcal{C}$ which is also idempotent.
\end{proof}

\begin{thm}
Let $\mathcal{C}=(1-u-v+uv)\mathcal{C}_{1}\oplus(u-uv)\mathcal{C}_{2}\oplus(v-uv)\mathcal{C}_{3}\oplus uv\mathcal{C}_{4}$ be a skew cyclic code of length $n$ over $\mathcal{R}$ and $gcd(n,k)=1$ and $gcd(n,q)=1$. Then $\mathcal{C}^{\perp}$ has an idempotent generator in $R[x;{\theta}_{t}]$.
\end{thm}
\begin{proof}
Suppose $e(x) = (1-u-v+uv)e_{1}(x)+(u-uv)e_{2}(x)+(v-uv)e_{3}(x)+uve_{4}(x)$ is an idempotent generator of $\mathcal{C}$ where $e_{i}(x)$ is idempotent generator of $\mathcal{C}_{i}$ for $i=1,2,3,4.$ Then $\mathcal{C}_{i}^{\perp}$ has idempotent generator
$1-e_{i}(x^{-1})$, for $i=1,2,3,4$ [See Lemma 12.3.23 (i) in \cite{book}].
Hence, by Theorem \ref{sk th2}, $\mathcal{C}^{\perp}$ has idempotent generator $(1-u-v+uv)(1-e_{1}(x^{-1}))+(u-uv)(1-e_{2}(x^{-1}))+(v-uv)(1-e_{3}(x^{-1}))+uv(1-e_{4}(x^{-1}))=1-e(x^{-1}).$
\end{proof}

\begin{exam}
Consider the field $F_{25}=F_{5}[\alpha]$, where ${\alpha}^{2}+\alpha +1=0$. Take $n=4$ and Frobenius automorphism ${\theta}_{t}:F_{25}\rightarrow F_{25}$ defined by ${\theta}_{t}(\alpha)={\alpha}^{5}.$ Now, $x^{4}-1=(x+2)(x+3)(x+\alpha)(x+\alpha+1)=(x+2)(x+3)(x+\alpha+1)(x+\alpha).$ Take $f_{1}(x)=f_{2}(x)=f_{3}(x)=f_{4}(x)=(x+\alpha+1)$. Then $\mathcal{C}_{1}=\langle f_{1}(x) \rangle, \mathcal{C}_{2}=\langle f_{2}(x) \rangle, \mathcal{C}_{3}=\langle f_{3}(x) \rangle, \mathcal{C}_{4}=\langle f_{4}(x) \rangle$ are skew cyclic codes of length 4 with dimension 3 over $F_{25}$. Let $f(x) =(1-u-v+uv)f_{1}(x)+(u-uv)f_{2}(x)+(v-uv)f_{3}(x)+uvf_{4}(x)=(x+\alpha+1).$ Then $\mathcal{C}=\langle f(x) \rangle$ is a skew cyclic code of length 4 over $\mathcal{R}=F_{25}+uF_{25}+vF_{25}+uvF_{25},$ where $u^{2}=u, v^{2}=v, uv=vu.$ Here, the Gray image $\Psi(\mathcal{C})$ is a skew quasi-cyclic code of index $4$ over $F_{25}$ with parameters $[16,12,2]$, which is an optimal code.
\end{exam}
\begin{exam}
Consider the field $F_{9}=F_{3}[2\alpha+1]$; where ${\alpha}^{2}+1=0$. Take $n=6$ and Frobenius automorphism ${\theta}_{t}:F_{9}\rightarrow F_{9}$ defined by ${\theta}_{t}(\alpha)={\alpha}^{3}$. Now, $x^{6}-1=(2+x+(1+2\alpha)x^{2}+x^{3})(1+x+(2\alpha+2)x^{2}+x^{3})=(2+\alpha x+2\alpha x^{3}+x^{4})(1+\alpha x+x^{2})$. Take $f_{1}(x)=f_{2}(x)=f_{3}(x)=(2+\alpha x+2\alpha x^{3}+x^{4})$ and $f_{4}(x)=(2+x+(1+2\alpha)x^{2}+x^{3})$, then $\mathcal{C}_{1}=\langle f_{1}(x) \rangle, \mathcal{C}_{2}=\langle f_{2}(x) \rangle, \mathcal{C}_{3}=\langle f_{3}(x) \rangle, \mathcal{C}_{4}=\langle f_{4}(x) \rangle$ are skew cyclic codes of length $6$ over $F_{9}$ with dimension $2, 2, 2$ and $3$ respectively. If we take $f(x) =(1-u-v+uv)f_{1}(x)+(u-uv)f_{2}(x)+(v-uv)f_{3}(x)+uvf_{4}(x)=(1-uv)(2+\alpha x+2\alpha x^{3}+x^{4})+uv(2+x+(1+2\alpha)x^{2}+x^{3})$, then $\mathcal{C}=\langle f(x) \rangle$ is a skew cyclic code of length 6 over $\mathcal{R}=F_{9}+uF_{9}+vF_{9}+uvF_{9}$, where $u^{2}=u, v^{2}=v, uv=vu.$ The Gray image $\Psi(\mathcal{C})$ is a skew quasi-cyclic code of index $4$ over $F_{9}$ with parameters $[24,9,4].$
\end{exam}

\section{Skew Constacyclic codes over $\mathcal{R}$}\label{sec6}
\begin{df}
Let $\alpha=(\alpha_{1}+u\alpha_{2}+v\alpha_{3}+uv\alpha_{4})$ be a unit in $\mathcal{R}$ where $\alpha_{i}\in F_{p^{t}}\backslash\{0\}$ and ${\theta}_{t}$ be the automorphism on $\mathcal{R}$. A linear code $\mathcal{C}$ of length $n$ is said to be skew $\alpha$-constacyclic code over $\mathcal{R}$ if and only if $\mathcal{C}$ is invariant under the skew $\alpha$-constacyclic shift operation ${\tau}_\alpha$ where $\tau_{\alpha}:\mathcal{R}^{n} \rightarrow \mathcal{R}^{n}$ defined by $\tau_{\alpha}(c_{0},c_{1},..,c_{n-1})=(\alpha \theta_{t}(c_{n-1}),\theta_{t}(c_{0}),..,\theta_{t}(c_{n-2}))$, $i.e, \mathcal{C}$ is skew $\alpha$-constacyclic code if and only if $\tau_{\alpha}(\mathcal{C})=\mathcal{C}.$\\
Clearly, $\mathcal{C}$ is said to be skew cyclic code for $\alpha= 1$ and skew negacyclic code for $\alpha= -1$.
\end{df}
In this section our study restricted to the condition $\alpha\in \mathcal{R}$ such that $\alpha^{2}=1$. By identifying each codeword $c = (c_{0}, c_{1}, \dots c_{n-1})\in \mathcal{R}^{n}$ to a polynomial $c(x) = c_{0}+c_{1}x+\dots +c_{n-1}x^{n-1}$ in the left $\mathcal{R}$-module $\mathcal{R}_{n}= \mathcal{R}[\theta_{t};x]/\langle x^{n}-\alpha\rangle$, we call a linear code $C$ is a skew $\alpha$-constacyclic code if and only if it is a left $\mathcal{R}$-submodule of $\mathcal{R}_{n}= \mathcal{R}[\theta_{t};x]/\langle x^{n}-\alpha\rangle$.

\begin{thm}
Define a map $\rho : \mathcal{R}_{n} = \mathcal{R}[\theta_{t};x]/\langle x^{n}-1 \rangle \mapsto \mathcal{R}_{n},\alpha = \mathcal{R}[\theta_{t};x]/\langle x^{n}-\alpha\rangle$ by $\rho(f(x)) = f(\alpha x)$. If $n$ is odd, then $\rho$ is a left $\mathcal{R}$-module isomorphism.
\end{thm}
\begin{proof}
Justification is straightforward. Only need to observe that if
\begin{align*}
f(x)&=g(x) ~mod~ (x^{n}-1)\\
\Leftrightarrow f(x)-g(x)&=h(x)\ast(x^{n}-1) ~for~ some~ h(x)\in \mathcal{R}[x;\theta_{t}]\\
\Leftrightarrow f(\alpha x)-g(\alpha x) &=h(\alpha x)\ast (\alpha^{n}x^{n}-1)\\
& =h(\alpha x)\ast (\alpha x^{n}-\alpha^{2})(as~\alpha^{n}=\alpha~for~odd~n)\\
& =\alpha h(\alpha x)\ast ( x^{n}-\alpha )\\
\Leftrightarrow f(\alpha x) &= g(\alpha x) ~mod~ (x^{n}-\alpha).
\end{align*}
\end{proof}

\begin{cor}
There is a one-to-one correspondence between the skew cyclic codes and skew $\alpha$-constacyclic codes over $\mathcal{R}$ of odd length.
\end{cor}

\begin{cor}
Let $n$ be an odd positive integer. Define a permutation $\mu$ on $\mathcal{R}^{n}$ as $\mu(c_{0}, c_{1}, \dots ,c_{n-1}) = (c_{0}, \alpha c_{1}, \dots ,\alpha^{n-1}c_{n-1})$. Then $\mathcal{C}$ is a skew cyclic code of length $n$ if and only if $\mu(\mathcal{C})$ is skew $\alpha$-constacyclic code of length $n$ over $\mathcal{R}$.
\end{cor}

\subsection{Relations}

\begin{df}
Let $\mathcal{C}$ be a linear code of length $n$ over $\mathcal{R}$ and $n=ml.$ Then $\mathcal{R}$ is said to be an $\alpha$-quasi-twisted code if for any
\begin{align*}
(c_{0 ~0},c_{0 ~1},...,c_{0 ~l-1},...,c_{m-1~ 0},c_{m-1 ~1},...,c_{m-1 ~l-1})\in \mathcal{C}
\end{align*}
implies
\begin{align*}
(\alpha c_{m-1 ~0},\alpha c_{m-1 ~1},...,\alpha c_{m-1~ l-1},...,c_{m-2 ~0 },c_{m-2 ~1},...,c_{m-2 ~l-1})\in \mathcal{C}.
\end{align*}
If $l$ be the least positive integer satisfying $n=ml$, then $\mathcal{C}$ is known as $\alpha$-quasi-twisted code of length $n$ over $\mathcal{R}.$
\end{df}

There is a nice relationship among skew $\alpha$-constacyclic codes, constacyclic codes and $\alpha$-quasi-twisted code over $\mathcal{R}$ which are obtain from following results.
\begin{thm}\label{ret 1}
Let $\mathcal{C}$ be a skew $\alpha$-constacyclic code of length $n$ over $\mathcal{R}$ and $gcd(n, k) = 1$. Then $\mathcal{C}$ is a $\alpha$-constacyclic code of length $n$ over $\mathcal{R}$.
\end{thm}
\begin{proof}
Since $gcd(n, k) = 1$, by elementary concept of number theory, there exists an integer $L>0$ such that $Tk= 1+Ln$.
Let $c(x)=c_{0}+c_{1}x+\dots +c_{n-1}x^{n-1}\in \mathcal{C}$. As $\mathcal{C}$ is a skew $\alpha$-constacyclic code and $x\ast c(x)$ represents skew $\alpha$-constacyclic shift of the codeword $c(x)$, so $x\ast c(x), x^{2}\ast c(x),...,x^{Tk}\ast c(x)$ are belong to $\mathcal{C}$, where
\begin{align*}
x^{Tk}\ast c(x)& = x^{Tk}\ast (c_{0}+c_{1}x+\dots +c_{n-1}x^{n-1})\\
& = \theta_{t}^{Tk}(c_{0})x^{Tk}+ \theta_{t}^{Tk}(c_{1})x^{Tk+1}+\dots + \theta_{t}^{Tk}(c_{n-1})x^{Tk+n-1}\\
& = c_{0}x^{1+Ln}+ c_{1}x^{2+Ln}+\dots + c_{n-1}x^{Ln+n}\\
& = \alpha^{L}(c_{0}x+c_{1}x^{2}+\dots + c_{n-2}x^{n-1}+\alpha c_{n-1})~(as~in~\mathcal{R}_{n},~ x^{n} = \alpha)\\
\implies \alpha^{L}x^{Tk}\ast c(x)& = c_{0}x+c_{1}x^{2}+\dots + c_{n-2}x^{n-1}+\alpha c_{n-1}\in \mathcal{C}~(as~ \alpha^{2} = 1).
\end{align*}
This proves that $\mathcal{C}$ is a $\alpha$-constacyclic code of length $n$ over $\mathcal{R}$.
\end{proof}

\begin{cor}
Let $gcd(n, k) = 1$. If $f(x)$ is a right divisor of $x^{n}-\alpha$ in the skew polynomial ring $\mathcal{R}[x;\theta_{t}]$, then $f(x)$ is a factor of $x^{n}-\alpha$ in the polynomial ring $\mathcal{R}[x]$.
\end{cor}

\begin{thm}\label{ret 2}
Let $\mathcal{C}$ be a skew $\alpha$-constacyclic code of length $n$ and $gcd(n,k) = l.$ Then $\mathcal{C}$ is a $\alpha$-quasi-twisted code of index $l$ over $\mathcal{R}.$
\end{thm}
\begin{proof}
 Since $gcd(n,k) = l$, there exit two integers $T$ and $D$ such that $Tk = l + Dn ; D>0.$
Let $r = (c_{0 ~0},c_{0~ 1},...,c_{0 ~l-1},...,c_{m-1~ 0},c_{m-1~ 1},...,c_{m-1~ l-1})\in \mathcal{C}$. Since $\mathcal{C}$ is a skew $\alpha$-constacyclic code, $\tau_{\alpha}(r), \tau_{\alpha}^{2}(r),...,\tau_{\alpha}^{l}(r)$  belong to $\mathcal{C}$, where
\begin{align*}
\tau_{\alpha}^{l}(r)&= (\theta_{t}^{l}(\alpha c_{m-1~ 0}),...,\theta_{t}^{l}(\alpha c_{m-1~ l-1}),\theta_{t}^{l}(c_{0~ 0}),...\theta_{t}^{l}(c_{0 ~l-1}),...,\\& ~~~~\theta_{t}^{l}(c_{m-2~ 0}),...,\theta_{t}^{l}(c_{m-2 ~l-1})).\\
\implies \tau_{\alpha}^{l+Dn}(r)&=(\theta_{t}^{l+Dn}(c_{m-1~ 0}),...,\theta_{t}^{l+Dn}(c_{m-1~ l-1}),...,\\& ~~~~\theta_{t}^{l+Dn}(\alpha c_{m-2~ 0}),...,\theta_{t}^{l+Dn}(\alpha c_{m-2~ l-1}))\\
&= (\theta_{t}^{Tk}(c_{m-1 ~0}),...,\theta_{t}^{Tk}(c_{m-1 ~l-1}),.\\&~~~~..,\theta_{t}^{Tk}(\alpha c_{m-2~ 0}),...,\theta_{t}^{Tk}(\alpha c_{m-2~ l-1}))\\
&= (c_{m-1~ 0},...,c_{m-1~ l-1},...,\alpha c_{m-2~ 0},...,\alpha c_{m-2~ l-1}).\\
\implies \alpha \tau_{\alpha}^{l+Dn}(r)&= (\alpha c_{m-1~ 0},...,\alpha c_{m-1~ l-1},...,c_{m-2 ~0},...,c_{m-2 ~l-1})\in \mathcal{C}~(as~ \alpha^{2} = 1)
\end{align*}
This proves that $\mathcal{C}$ is an $\alpha$-quasi-twisted code of index $l$.
\end{proof}

\section{Constacyclic code with other shift constant}\label{sec7}

In this section we characterize the Gray images of skew $\alpha$-constacyclic codes over $\mathcal{R}$ as a skew quasi-twisted code over $F_{q}$. We define  the quasi-twisted shift operator $\omega_{l}$ on $(F_{q}^{n})^{l}$ by \\
\begin{align*}
\omega_{l}((c^{1})\mid (c^{2})\mid(c^{3})\mid\dots\mid(c^{l})) = (\tau_{\alpha}(c^{1})\mid \tau_{\alpha}(c^{2})\mid\tau_{\alpha}(c^{3})\mid\dots \mid\tau_{\alpha}(c^{l}))
\end{align*}
where $c^{i}\in F_{q}^{n}$ and $\tau_{\alpha}$ is skew $\alpha$-constacyclic shift operators as define in last section \ref{sec6}.\\
A linear code $\mathcal{C}$ of length $nl$ over $F_{q}$ is said to be a skew quasi-twisted code of index $l$ if $\omega_{l}(C) = C$.

With the help of above definition, we get the following results:

\begin{pro}\label{pro3}
Let $\Psi$ be the Gray map as define earlier, then $\Psi\tau_{\alpha} = \omega_{4}\Psi$.
\end{pro}
\begin{proof}
Let $r = (r_{0}, r_{1},\dots ,r_{n-1})\in \mathcal{R}$. We have,
$\Psi\tau_{\alpha}(r) = \Psi(\alpha \theta_{t}(r_{n-1}), \theta_{t}(r_{0}),\\ \dots , \theta_{t}(r_{n-2})) = (\alpha a_{n-1}^{p^t}, a_{0}^{p^t},\dots , a_{n-2}^{p^t}, \alpha a_{n-1}^{p^t}+\alpha b_{n-1}^{p^t},\dots , a_{n-2}^{p^t}+b_{n-2}^{p^t}, \alpha a_{n-1}^{p^t}+\alpha c_{n-1}^{p^t}, \dots , a_{n-2}^{p^t}+ c_{n-2}^{p^t}, \alpha a_{n-1}^{p^t}+\alpha b_{n-1}^{p^t}+\alpha c_{n-1}^{p^t}+\alpha d_{n-1}^{p^t}, \dots , a_{n-2}^{p^t}+b_{n-2}^{p^t}+c_{n-2}^{p^t}+d_{n-2}^{p^t}).$\\ \\
On the other hand,\\ \\
$\omega_{4}\Psi(r) = \omega_{4}(a_{0}, a_{1},\dots ,a_{n-1}, a_{0}+b_{0}, a_{1}+b_{1}, \dots , a_{n-1}+b_{n-1}, a_{0}+c_{0}, a_{1}+c_{1}, \dots , a_{n-1}+c_{n-1}, a_{0}+b_{0}+c_{0}+d_{0}, a_{1}+b_{1}+c_{1}+d_{1}, \dots , a_{n-1}+b_{n-1}+c_{n-1}+d_{n-1}) = (\alpha a_{n-1}^{p^t}, a_{0}^{p^t},\dots , a_{n-2}^{p^t}, \alpha a_{n-1}^{p^t}+\alpha b_{n-1}^{p^t},\dots , a_{n-2}^{p^t}+b_{n-2}^{p^t}, \alpha a_{n-1}^{p^t}+\alpha c_{n-1}^{p^t}, \dots , a_{n-2}^{p^t}+ c_{n-2}^{p^t}, \alpha a_{n-1}^{p^t}+\alpha b_{n-1}^{p^t}+\alpha c_{n-1}^{p^t}+\alpha d_{n-1}^{p^t}, \dots , a_{n-2}^{p^t}+b_{n-2}^{p^t}+c_{n-2}^{p^t}+d_{n-2}^{p^t}).$ Therefore, $\Psi\tau_{\alpha} = \omega_{4}\Psi$.
\end{proof}

As a consequence of the Proposition \ref{pro3}, we have the following:

\begin{thm}
If $C$ is a skew $\alpha$-constacyclic codes of length $n$ over $\mathcal{R}$, then its $F_{q}$-image $\Psi(C)$ is a skew quasi-twisted code of length $4n$ over $F_{q}$ of index $4$ and vice versa.
\end{thm}

\section{ Decomposition of skew $\boldsymbol{(\alpha_{1}+u\alpha_{2}+v\alpha_{3}+uv\alpha_{4})}$-constacyclic codes over $\mathcal{R}$} \label{sec8}
In this section, we discuss skew $(\alpha_{1}+u\alpha_{2}+v\alpha_{3}+uv\alpha_{4})$-constacyclic codes of arbitrary length $n$ over $\mathcal{R}$ by decomposition method.\\

Gao et al. \cite{gao} have shown that a skew $\alpha$-constacyclic code $\mathcal{C}$ of length $n$ over $F_{q}$ is a left $F_{q}[x;\theta_{t}]$-submodule of $F_{q}[x;\theta_{t}]/\langle x^{n}-\alpha \rangle$ generated by a monic polynomial $f(x)$ with minimal degree in $\mathcal{C}$ and $f(x)$ is a right divisor of $(x^{n}-\alpha)$.
Motivated by these study we find some structural properties of skew $\alpha$-constacyclic codes over $\mathcal{R}$ which are listed below.
\begin{thm}\label{de th1}
Let $\mathcal{C}=(1-u-v+uv)\mathcal{C}_{1}\oplus(u-uv)\mathcal{C}_{2}\oplus(v-uv)\mathcal{C}_{3}\oplus uv\mathcal{C}_{4}$ be a linear code of length $n$ over $\mathcal{R}$. Then $\mathcal{C}$ is $(\alpha_{1}+u\alpha_{2}+v\alpha_{3}+uv\alpha_{4})$-constacyclic code over $\mathcal{R}$ if and only if $\mathcal{C}_{1},\mathcal{C}_{2},\mathcal{C}_{3}$ and $\mathcal{C}_{4}$ are skew $\alpha_{1}$-constacyclic code, skew $(\alpha_{1}+\alpha_{2})$-constacyclic code, skew $(\alpha_{1}+\alpha_{3})$-constacyclic code and skew $(\alpha_{1}+\alpha_{2}+\alpha_{3}+\alpha_{4})$-constacyclic code of length $n$ over $F_{q}$ respectively, where $(\alpha_{1}+u\alpha_{2}+v\alpha_{3}+uv\alpha_{4})$ is a unit in $\mathcal{R}.$
\end{thm}
\begin{proof}
Let $r=(r_{0},r_{1},..,r_{n-1})\in \mathcal{C}$, where $r_{i}=(1-u-v+uv)a_{i}+(u-uv)b_{i}+(v-uv)c_{i}+uvd_{i},0\leq i\leq n-1$. Take $a=(a_{0},..,a_{n-1}),b=(b_{0},..,b_{n-1}),c=(c_{0},..,c_{n-1}),d=(d_{0},..,d_{n-1})$, then $a\in \mathcal{C}_{1},b\in \mathcal{C}_{2},c\in \mathcal{C}_{3}$ and $d\in \mathcal{C}_{4}$. Suppose $\mathcal{C}_{1},\mathcal{C}_{2},\mathcal{C}_{3}$ and $\mathcal{C}_{4}$ are skew $\alpha_{1}$-constacyclic, skew $(\alpha_{1}+\alpha_{2})$-constacyclic, skew $(\alpha_{1}+\alpha_{3})$-constacyclic, skew $(\alpha_{1}+\alpha_{2}+\alpha_{3}+\alpha_{4})$-constacyclic code over $F_{q}$ respectively. So $\tau_{\alpha_{1}}(a)\in \mathcal{C}_{1},\tau_{\alpha_{1}+\alpha_{2}}(b)\in \mathcal{C}_{2},\tau_{\alpha_{1}+\alpha_{3}}(c)\in \mathcal{C}_{3}$, and  $\tau_{\alpha_{1}+\alpha_{2}+\alpha_{3}+\alpha_{4}}(d)\in \mathcal{C}_{4}.$ Now, $\tau_{\alpha_{1}+u\alpha_{2}+v\alpha_{3}+uv\alpha_{4}}(r)$=
$((\alpha_{1}+u\alpha_{2}+v\alpha_{3}+uv\alpha_{4})\theta_{t}(r_{n-1}),\theta_{t}(r_{0}),..,\theta_{t}(r_{n-2}))=(1-u-v+uv)\tau_{\alpha_{1}}(a)+(u-uv)\tau_{\alpha_{1}+\alpha_{2}}(b)+(v-uv)\tau_{\alpha_{1}+\alpha_{3}}(c)+uv\tau_{\alpha_{1}+\alpha_{2}+\alpha_{3}+\alpha_{4}}(d)\in (1-u-v+uv)\mathcal{C}_{1}\oplus(u-uv)\mathcal{C}_{2}\oplus(v-uv)\mathcal{C}_{3}\oplus uv\mathcal{C}_{4}=\mathcal{C}.$ This shows that $\mathcal{C}$ is $(\alpha_{1}+u\alpha_{2}+v\alpha_{3}+uv\alpha_{4})$-constacyclic code over $\mathcal{R}$.\\
Conversely, let $a=(a_{0},..,a_{n-1})\in \mathcal{C}_{1},b=(b_{0},..,b_{n-1})\in \mathcal{C}_{2},c=(c_{0},..,c_{n-1})\in \mathcal{C}_{3},d=(d_{0},..,d_{n-1})\in \mathcal{C}_{4}$. Take $r_{i}=(1-u-v+uv)a_{i}+(u-uv)b_{i}+(v-uv)c_{i}+uvd_{i},0\leq i\leq n-1$. Then $r=(r_{0},r_{1},..,r_{n-1})\in \mathcal{C}$. Suppose $\mathcal{C}$ is $(\alpha_{1}+u\alpha_{2}+v\alpha_{3}+uv\alpha_{4})$-constacyclic code over $\mathcal{R}$. So, $\tau_{\alpha_{1}+u\alpha_{2}+v\alpha_{3}+uv\alpha_{4}}(r)\in \mathcal{C}$, where $\tau_{\alpha_{1}+u\alpha_{2}+v\alpha_{3}+uv\alpha_{4}}(r)=(1-u-v+uv)\tau_{\alpha_{1}}(a)+(u-uv)\tau_{\alpha_{1}+\alpha_{2}}(b)+(v-uv)\tau_{\alpha_{1}+\alpha_{3}}(c)+uv\tau_{\alpha_{1}+\alpha_{2}+\alpha_{3}+\alpha_{4}}(d).$ It follows that $\tau_{\alpha_{1}}(a)\in \mathcal{C}_{1},\tau_{\alpha_{1}+\alpha_{2}}(b)\in \mathcal{C}_{2},\tau_{\alpha_{1}+\alpha_{3}}(c)\in \mathcal{C}_{3}$ and  $\tau_{\alpha_{1}+\alpha_{2}+\alpha_{3}+\alpha_{4}}(d)\in \mathcal{C}_{4}.$ Hence $\mathcal{C}_{1},\mathcal{C}_{2},\mathcal{C}_{3}$ and $\mathcal{C}_{4}$ are skew $\alpha_{1}$-constacyclic code, skew $(\alpha_{1}+\alpha_{2})$-constacyclic code, skew $(\alpha_{1}+\alpha_{3})$-constacyclic code and skew $(\alpha_{1}+\alpha_{2}+\alpha_{3}+\alpha_{4})$-constacyclic code of length $n$ over $F_{q}$ respectively.
\end{proof}

\begin{cor}
Let $\mathcal{C}=(1-u-v+uv)\mathcal{C}_{1}\oplus(u-uv)\mathcal{C}_{2}\oplus(v-uv)\mathcal{C}_{3}\oplus uv\mathcal{C}_{4}$ be a skew $(\alpha_{1}+u\alpha_{2}+v\alpha_{3}+uv\alpha_{4})$-constacyclic code of length $n$ over $\mathcal{R}$. Then the dual code $\mathcal{C}^{\perp}=(1-u-v+uv)\mathcal{C}_{1}^{\perp}\oplus(u-uv)\mathcal{C}_{2}^{\perp}\oplus(v-uv)\mathcal{C}_{3}^{\perp}\oplus uv\mathcal{C}_{4}^{\perp}$ is skew $(\alpha_{1}+u\alpha_{2}+v\alpha_{3}+uv\alpha_{4})^{-1}$-constacyclic code over $\mathcal{R}$, where $\mathcal{C}_{1}^{\perp},\mathcal{C}_{2}^{\perp},\mathcal{C}_{3}^{\perp},\mathcal{C}_{4}^{\perp}$ are skew $\alpha_{1}^{-1}$-constacyclic code, skew $(\alpha_{1}+\alpha_{2})^{-1}$-constacyclic code, skew $(\alpha_{1}+\alpha_{3})^{-1}$-constacyclic code and skew $(\alpha_{1}+\alpha_{2}+\alpha_{3}+\alpha_{4})^{-1}$-constacyclic code over $F_{q}$ respectively, provided $n$ is a multiple of $k$(order of the automorphism).
\end{cor}
\begin{proof}
As $(\alpha_{1}+u\alpha_{2}+v\alpha_{3}+uv\alpha_{4})$ is fixed by $\theta_{t}$ and $n$ is a multiple of $k$, then by Lemma 3.1 of \cite{jitman}, $\mathcal{C}^{\perp}$ is $(\alpha_{1}+u\alpha_{2}+v\alpha_{3}+uv\alpha_{4})^{-1}$-constacyclic over $\mathcal{R}$. Since $(\alpha_{1}+u\alpha_{2}+v\alpha_{3}+uv\alpha_{4})=(1-u-v+uv)\alpha_{1}+(u-uv)(\alpha_{1}+\alpha_{2})+(v-uv)(\alpha_{1}+\alpha_{3})+uv(\alpha_{1}+\alpha_{2}+\alpha_{3}+\alpha_{4})$, it follows $(\alpha_{1}+u\alpha_{2}+v\alpha_{3}+uv\alpha_{4})^{-1}=(1-u-v+uv){\alpha_{1}}^{-1}+(u-uv)(\alpha_{1}+\alpha_{2})^{-1}+(v-uv)(\alpha_{1}+\alpha_{3})^{-1}+uv(\alpha_{1}+\alpha_{2}+\alpha_{3}+\alpha_{4})^{-1}.$ Hence by Theorem \ref{de th1}, $\mathcal{C}_{1}^{\perp},\mathcal{C}_{2}^{\perp},\mathcal{C}_{3}^{\perp},\mathcal{C}_{4}^{\perp}$ are skew $\alpha_{1}^{-1}$-constacyclic code, skew $(\alpha_{1}+\alpha_{2})^{-1}$-constacyclic code, skew $(\alpha_{1}+\alpha_{3})^{-1}$-constacyclic code and skew $(\alpha_{1}+\alpha_{2}+\alpha_{3}+\alpha_{4})^{-1}$-constacyclic code over $F_{q}$ respectively.
\end{proof}

\begin{cor}
Let $\mathcal{C}=(1-u-v+uv)\mathcal{C}_{1}\oplus(u-uv)\mathcal{C}_{2}\oplus(v-uv)\mathcal{C}_{3}\oplus uv\mathcal{C}_{4}$ be a skew $(\alpha_{1}+u\alpha_{2}+v\alpha_{3}+uv\alpha_{4})$-constacyclic code of length $n$ over $\mathcal{R}$. Then $\mathcal{C}$ is self-dual if and only if $\alpha_{1}+u\alpha_{2}+v\alpha_{3}+uv\alpha_{4}=1,-1,1-2u,1-2v,1-2uv,-1+2u,-1+2v,-1+2uv,1-2u+2uv,1-2v+2uv,-1+2u-2uv,-1+2v-2uv,1-2u-2v+2uv,-1+2u+2v-2uv,1-2u-2v+4uv,-1+2u+2v-4uv.$
\end{cor}
\begin{proof}
It can easily be proved that $\mathcal{C}$ is self-dual if and only if $\alpha_{1}=\pm 1, \alpha_{1}+\alpha_{2}=\pm 1, \alpha_{1}+\alpha_{2}+\alpha_{3}=\pm 1$ and $\alpha_{1}+\alpha_{2}+\alpha_{3}+\alpha_{4}=\pm 1.$
\end{proof}

\begin{thm}\label{de th2}
Let $\mathcal{C}=(1-u-v+uv)\mathcal{C}_{1}\oplus(u-uv)\mathcal{C}_{2}\oplus(v-uv)\mathcal{C}_{3}\oplus uv\mathcal{C}_{4}$ be a skew $(\alpha_{1}+u\alpha_{2}+v\alpha_{3}+uv\alpha_{4})$-constacyclic code over $\mathcal{R}$. Then there exits a polynomial $f(x)$ in ${\mathcal{R}[x;{\theta}_{t}]}$ which is a right divisor of $x^{n}-(\alpha_{1}+u\alpha_{2}+v\alpha_{3}+uv\alpha_{4})$ and $\mathcal{C}=\langle f(x)\rangle.$
\end{thm}
\begin{proof}
Let $f_{i}(x)$ be generator of $\mathcal{C}_{i}$ for $i=1,2,3,4$. Then $(1-u-v+uv)f_{1}(x),(u-uv)f_{2}(x),(v-uv)f_{3}(x)$ and $uvf_{4}(x)$ are generators of $\mathcal{C}$. Take $f(x)=(1-u-v+uv)f_{1}(x)+(u-uv)f_{2}(x)+(v-uv)f_{3}(x)+uvf_{4}(x)$ and $\mathcal{G}=\langle f(x) \rangle$. Then $\mathcal{G}\subseteq\mathcal{C}$. On the other hand, $(1-u-v+uv)f(x)=(1-u-v+uv)f_{1}(x)\in\mathcal{G}, (u-uv)f(x)=(u-uv)f_{2}(x)\in \mathcal{G} , (v-uv)f(x)=(v-uv)f_{3}(x)\in \mathcal{G}$ and $uvf(x)=uvf_{4}(x)\in \mathcal{G}$. This implies that $\mathcal{C}\subseteq \mathcal{G}$ and hence $\mathcal{C}=\mathcal{G}=\langle f(x) \rangle.$\\
Since $f_{1}(x), f_{2}(x), f_{3}(x)$ and $f_{4}(x)$ are right divisors of $x^{n}-\alpha_{1},x^{n}-(\alpha_{1}+\alpha_{2}), x^{n}-(\alpha_{1}+\alpha_{3})$, and $x^{n}-(\alpha_{1}+\alpha_{2}+\alpha_{3}+\alpha_{4})$ respectively, so there exist $h_{1}(x),h_{2}(x),h_{3}(x),h_{4}(x)$ such that $x^{n}-\alpha_{1}=h_{1}(x)\ast f_{1}(x),x^{n}-(\alpha_{1}+\alpha_{2})=h_{2}(x)\ast f_{2}(x),x^{n}-(\alpha_{1}+\alpha_{3})=h_{3}(x)\ast f_{3}(x)$ and $x^{n}-(\alpha_{1}+\alpha_{2}+\alpha_{3}+\alpha_{4})=h_{4}(x)\ast f_{4}(x).$ Also, $[(1-u-v+uv)h_{1}+(u-uv)h_{2}+(v-uv)h_{3}+uvh_{4}]\ast f(x)=(1-u-v+uv)h_{1}\ast f_{1}+(u-uv)h_{2}\ast f_{2}+(v-uv)h_{3}\ast f_{3}+uvh_{4}\ast f_{4}=x^{n}-(\alpha_{1}+u\alpha_{2}+v\alpha_{3}+uv\alpha_{4})$. This proves that $f(x)$ is a right divisor of $x^{n}-(\alpha_{1}+u\alpha_{2}+v\alpha_{3}+uv\alpha_{4})$.
\end{proof}

\begin{cor}
Each left submodule of $\mathcal{R}[x;\theta_{t}]/\langle x^{n}-(\alpha_{1}+u\alpha_{2}+v\alpha_{3}+uv\alpha_{4}) \rangle$ is generated by single element where $\alpha_{1}+u\alpha_{2}+v\alpha_{3}+uv\alpha_{4}$ is a unit in $\mathcal{R}$.
\end{cor}

\begin{cor}
Let $\mathcal{C}=(1-u-v+uv)\mathcal{C}_{1}\oplus(u-uv)\mathcal{C}_{2}\oplus(v-uv)\mathcal{C}_{3}\oplus uv\mathcal{C}_{4}$ be a skew $(\alpha_{1}+u\alpha_{2}+v\alpha_{3}+uv\alpha_{4})$-constacyclic code of length $n$ over $\mathcal{R}$ and $f_{1}(x), f_{2}(x), f_{3}(x)$ and $f_{4}(x)$ be generators of $\mathcal{C}_{1},\mathcal{C}_{2},\mathcal{C}_{3}$ and $\mathcal{C}_{4}$ in $F_{q}[x;\theta_{t}]$ respectively such that $x^{n}-\alpha_{1}=h_{1}(x)\ast f_{1}(x),x^{n}-(\alpha_{1}+\alpha_{2})=h_{2}(x)\ast f_{2}(x),x^{n}-(\alpha_{1}+\alpha_{3})=h_{3}(x)\ast f_{3}(x)$ and $x^{n}-(\alpha_{1}+\alpha_{2}+\alpha_{3}+\alpha_{4})=h_{4}(x)\ast f_{4}(x).$ Then $\mathcal{C}^{\perp}=\langle (1-u-v+uv)\widehat{h}_{1}(x)+(u-uv)\widehat{h}_{2}(x)+(v-uv)\widehat{h}_{3}(x)+uv\widehat{h}_{4}(x)\rangle$ and $\mid \mathcal{C}^{\perp}\mid =q^{\sum_{i=1}^{4}f_{i}(x)}$.
\end{cor}
\begin{proof}
 Same as Corollary \ref{cor 1}.
\end{proof}

\begin{thm}
Let $\mathcal{C}=(1-u-v+uv)\mathcal{C}_{1}\oplus(u-uv)\mathcal{C}_{2}\oplus(v-uv)\mathcal{C}_{3}\oplus uv\mathcal{C}_{4}$ be a skew $(\alpha_{1}+u\alpha_{2}+v\alpha_{3}+uv\alpha_{4})$-constacyclic code over $\mathcal{R}$. Let $gcd(n,k)=1$ and $gcd(n,q)=1$. Then there exists an idempotent generator $e(x)=(1-u-v+uv)e_{1}(x)+(u-uv)e_{2}(x)+(v-uv)e_{3}(x)+uve_{4}(x)$ in $\mathcal{R}[x;\theta_{t}]/\langle x^{n}-(\alpha_{1}+u\alpha_{2}+v\alpha_{3}+uv\alpha_{4}) \rangle$ such that $\mathcal{C}=\langle e(x) \rangle$, where $e_{1}(x) \in F_{q}[x; \theta_{t}]/ \langle x^{n}- \alpha_{1} \rangle, e_{2}(x)\in F_{q}[x;\theta_{t}]/ \langle x^{n}-(\alpha_{1}+\alpha_{2}) \rangle,e_{3}(x)\in F_{q}[x;\theta_{t}]/\langle x^{n}-(\alpha_{1}+\alpha_{3}) \rangle, e_{4}(x)\in F_{q}[x;\theta_{t}]/\langle x^{n}-(\alpha_{1}+\alpha_{2}+\alpha_{3}+\alpha_{4}) \rangle$ are idempotent generators of $\mathcal{C}_{1},\mathcal{C}_{2},\mathcal{C}_{3}$ and $\mathcal{C}_{4}$ respectively.
\end{thm}
\begin{proof}
By using same argument of the proof of Theorem 16 of \cite{siap11}, we conclude that there exist idempotent generators $e_{1}(x),e_{2}(x),e_{3}(x)$ and $e_{4}(x)$ of $\mathcal{C}_{1},\mathcal{C}_{2},\mathcal{C}_{3}$ and $\mathcal{C}_{4}$ respectively in $F_{q}[x; \theta_{t}]$. Then by Theorem \ref{de th2}, $e(x)=(1-u-v+uv)e_{1}(x)+(u-uv)e_{2}(x)+(v-uv)e_{3}(x)+uve_{4}(x)$ is an idempotent generator of $\mathcal{C}.$
\end{proof}

\begin{exam}
Consider the field $F_{49}=F_{7}[\alpha];$ where ${\alpha}^{2}-\alpha+3=0$. Take $n=4$ and Frobenius automorphism $\theta_{t}:F_{49}\rightarrow F_{49}$ defined by $\theta_{t}(\alpha)={\alpha}^{7}.$ Now, $x^{4}-1=(x+1)(x-1)(x^{2}+1); x^{4}+1=(x^{2}+3x+1)(x^{2}+4x+1)$. Take $f_{1}(x)=f_{2}(x)=f_{3}(x)=(x+1)$ and $f_{4}(x)=x^{2}+4x+1$. Then $\mathcal{C}=\langle (1-u-v+uv)f_{1}(x)+(u-uv)f_{2}(x)+(v-uv)f_{3}(x)+uvf_{4}(x) \rangle=\langle(1-uv)(x+1)+uv(x^{2}+4x+1)\rangle$ is a self-dual skew $(1-2uv)$-constacyclic code of length $4$ over $\mathcal{R}=F_{49}+uF_{49}+vF_{49}+uvF_{49}$, where $u^{2}=u,v^{2}=v,uv=vu.$ Also, $\mid\langle\theta_{t}\rangle\mid = k = 2$ and $gcd(n, k) = 2$. So, by Theorem \ref{ret 2}, $C$ is a $(1-2uv)$-quasi-twisted code of length $4$ over $\mathcal{R}$ of index 2.
\end{exam}

\begin{exam}
Consider the field $F_{9} = F_{3}[2\alpha+1];$ where ${\alpha}^{2}+1=0$. Take Frobenius automorphism $\theta_{t}:F_{9}\rightarrow F_{9}$ defined by $\theta_{t}(\alpha)={\alpha}^{3}$ and $\mathcal{R}=F_{9}+uF_{9}+vF_{9}+uvF_{9}$, where $u^{2}=u,v^{2}=v,uv=vu.$ The polynomial $f(x) = x^{6}+(1-2v)x^{5}+x^{4}+(1-2v)x^{3}+x^{2}+(1-2v)x+1$ is a right divisor of $x^{7}-(1-2v-2uv)$ in $\mathcal{R}[\theta_{t};x]$ and also  $gcd(n, k) = gcd(7, 2) = 1$. Therefore, by Theorem \ref{ret 1}, $C = \langle f(x) \rangle$ is a $(1-2v-2uv)$-constacyclic code of length 7 over $\mathcal{R}$.
\end{exam}

\section{Conclusion}\label{sec9}
In this paper, we considered skew cyclic and skew constacyclic codes over $\mathcal{R}=F_{p^{m}}+uF_{p^{m}}+vF_{p^{m}}+uvF_{p^{m}}$, where $p$ is an odd prime and $u^{2}=u,v^{2}=v,uv=vu.$ It is shown that skew cyclic codes and skew constacyclic codes over $\mathcal{R}$ are principally generated. Also, we have given the necessary and sufficient conditions for codes being self-dual over $\mathcal{R}$.

\section*{Acknowledgement}
The authors are thankful to University Grant Commission (UGC), Government of India for financial support under Ref. No. 20/12/2015(ii)EU-V dated 31/08/2016 and Indian Institute of Technology Patna for providing the research facilities. The authors would like to thank the anonymous referees for their useful comments and suggestions.

\end{document}